\documentclass[journal]{IEEEtran}
\usepackage{amsmath,mathtools,graphicx,amsthm,amssymb,float,color,array,epsfig,subcaption}

\newcommand{\avec}{{\bf{a}}}

\newcommand{\bvec}{{\bf{b}}}

\newcommand{\yvec}{{\bf{y}}}

\newcommand{\wvec}{{\bf{w}}}
\newcommand{\xvec}{{\bf{x}}}

\newcommand{\rvec}{{\bf{r}}}

\newcommand{\hvec}{{\bf{h}}}

\newcommand{\onevec}{{\bf{1}}}
\newcommand{\zerovec}{{\bf{0}}}

\newcommand{\muvec}{{\bf{\mu}}}
\newcommand{\thetavec}{{\bf{\theta}}}

\newcommand{\Amat}{{\bf{A}}}

\newcommand{\Cmat}{{\bf{C}}}
\newcommand{\Dmat}{{\bf{D}}}

\newcommand{\Gmat}{{\bf{G}}}
\newcommand{\Hmat}{{\bf{H}}}

\newcommand{\Imat}{{\bf{I}}}

\newcommand{\Rmat}{{\bf{R}}}

\newcommand{\Wmat}{{\bf{W}}}

\newcommand{\define}{\stackrel{\triangle}{=}}

\newcommand{\Phimat}{\mbox{\boldmath $\Phi$}}

\def\Sigmavec{{\mbox{\boldmath $\Sigma$}}}

\def\thetavec{{\mbox{\boldmath $\theta$}}}

\def\muvec{{\mbox{\boldmath $\mu$}}}

\def\thetavecsmall{{\mbox{\boldmath {\scriptsize $\theta$}}}}

\newcommand{\be}{\begin{equation}}
\newcommand{\ee}{\end{equation}}
\newcommand{\beqna}{\begin{eqnarray}}
\newcommand{\eeqna}{\end{eqnarray}}

\usepackage{enumerate}
\usepackage{float}

\usepackage{multirow}
\usepackage{cite}
\usepackage{graphicx,nicefrac}
\usepackage[ruled,vlined]{algorithm2e}

\newtheorem{theorem}{Theorem}
\newtheorem{proposition}{Proposition}

\usepackage{xcolor}

\begin{document}
\title{Resource Allocation and Dithering of Bayesian Parameter Estimation Using Mixed-Resolution Data}
\author{Itai~E.~Berman,~\IEEEmembership{Student Member, IEEE} and 
Tirza Routtenberg, \IEEEmembership{Senior Member, IEEE}
\thanks{This work is partially supported by  the ISRAEL SCIENCE FOUNDATION (ISF), grant No. 1173/16.}
\thanks{{\footnotesize{I. Berman and T. Routtenberg are with the School of Electrical and Computer Engineering Ben-Gurion University of the Negev Beer-Sheva 84105, Israel, e-mail: itaieliy@post.bgu.ac.il, tirzar@bgu.ac.il.}}}
}

\maketitle

\begin{abstract}
Quantization of signals is an integral part of modern signal processing applications, such as sensing, communication, and inference. While signal quantization provides many physical advantages, it usually degrades the subsequent estimation performance that is based on quantized data. In order to maintain physical constraints and simultaneously bring substantial performance gain, in this work we consider systems with mixed-resolution, 1-bit quantized and continuous-valued, data.
First, we describe the linear minimum mean-squared error (LMMSE) estimator and its associated mean-squared error (MSE) for the general mixed-resolution model.
However, the MSE of the LMMSE requires matrix inversion in which the number of measurements defines the matrix dimensions  and thus, is not a tractable tool for optimization and system design. Therefore, we present
the linear Gaussian orthonormal (LGO) measurement model and derive a closed-form analytic expression for the MSE of the LMMSE estimator under this model. In addition, we present two common special cases of the LGO model: 1) scalar parameter estimation and 2) channel estimation in mixed-ADC multiple-input multiple-output (MIMO) communication systems.
We then solve the resource allocation optimization problem of the LGO model with the proposed tractable form of the MSE as an objective function and under a power constraint using a one-dimensional search.
Moreover, we present the concept of dithering for mixed-resolution models and optimize the dithering noise as part of the resource allocation optimization problem for two dithering schemes: 1) adding noise only to the quantized measurements and 2) adding noise to both measurement types.
Finally, we present simulations that demonstrate the advantages of using mixed-resolution measurements and the possible improvement introduced with dithering and resource allocation.
\end{abstract}


\IEEEpeerreviewmaketitle

\section{Introduction} \label{I}
Traditional statistical signal processing (SSP) techniques were developed for high-resolution sensors,
under the  unrealistic assumption of infinite precision sampling, or ``analog" data,
that can neglect the quantization effect. High-resolution sensors result in high  performance in various SSP tasks, such as parameter estimation.
In modern signal processing, signal quantization plays an important role 
with   various applications, including  wireless sensor networks (WSNs)
\cite{Ribeiro_Giannakis_2006,Ribeiro_Giannakis2006P2,chavali2012managing,Saska_Blum_Kaplan2015,mandic2008signal}, direction of arrival estimation \cite{corey2017wideband}, target tracking \cite{Ribeiro_Giannakis_2010,7163356,4524043}, multiple-input multiple-output (MIMO) communications \cite{li2017channel,park2017optimization,zhang2016mixed,liang2016mixed,Choi_Heath_2016,pirzadeh2018spectral,shlezinger2019asymptotic,marzetta2010noncooperative}, cognitive radio \cite{Lunden_Koivunen_Poor2015}, and array processing \cite{bar_Weiss_2002}.
For example, in communication systems there is usually a need for cheaper, less power-hungry analog-to-digital converters (ADCs) while maintaining accurate channel estimation \cite{liang2016mixed}.
The widespread use of signal quantization is due to its many advantages, which include reduction of hardware complexity, power consumption, communication bandwidth, sensor cost, and sensor's physical dimensions, as well as enabling  high-rate sampling \cite{walden1999analog,1550190}.
Despite all its practical advantages, quantization results in low-resolution signals, which degrades the performance of subsequent parameter estimation that is based on the quantized data.
Moreover, signal quantization introduces nonlinear effects into the system, which poses new challenges for parameter estimation that relies on these signals, such as non-convex optimizations \cite{Papadopoulos_Wornell_Oppenheim_2001,Stein_Bar_2018}.

The nonlinear problem of parameter estimation based on low-precision samples, especially from 1-bit (signed) measurements has been discussed widely in the literature (see, e.g. \cite{Stein_Bar_2018,Madsen_Handel_2000,Kipnis_2019,8392734,stein2019spectral}).
For instance, in \cite{Ribeiro_Giannakis_2006,Ribeiro_Giannakis2006P2,7111345}
the maximum-likelihood (ML) estimator and the corresponding Cram$\acute{\text{e}}$r-Rao bounds (CRBs) for quantized samples  are presented. However, the ML is usually intractable, and, thus,
various suboptimal, low-complexity methods have been developed in the literature \cite{Choi_Heath_2016,Ren_Stoica_2019,bar_Weiss_2002,Madsen_Handel_2000}.
In \cite{zeitler2012bayesian} the problem of parameter estimation of a random parameter using 1-bit dithered measurements is studied, deriving lower bounds on the mean-squared error (MSE) using the Bayesian CRB and designing dither strategies.
Studies on channel estimation in massive MIMO systems with 1-bit ADCs show acceptable performance in channel capacity and the achievable rate due to the use of a large number of antennas compared to that of analog ADCs \cite{li2017channel,jacobsson2015one,wan2020generalized}.
In all these methods the use of quantized data results in a degradation of the estimation performance compared with the analog-data based methods.

In addition to purely-quantized or purely-analog data, a few works have been using schemes with multiple quantization resolution data \cite{liang2016mixed,harel2017non,Saska_Blum_Kaplan2015}. 
For example, the ML and CRB for non-Bayesian estimation  with partially quantized observations is considered in \cite{harel2017non}. For the Bayesian case, the minimum MSE (MMSE) estimation of a uniformly distributed parameter, based on both quantized and unquantized observations, has been suggested in \cite{Saska_Blum_Kaplan2015} and linear MMSE (LMMSE) in specific applications is discussed in
\cite{park2017optimization,zhang2016mixed,pirzadeh2018spectral}.
However, while there are various estimation algorithms based on quantized signals (see above),
there has been less emphasis on the analysis and design of mixed-resolution
architectures. Considerable improvements could be obtained by optimization of 
estimation schemes relying on both quantized and continuous-valued data with respect to their parameters, performance, and complexity. This optimization is crucial  in order   to cope with the limited resources for data processing, storage, and communication  in real-world applications.
However,  incorporating quantized data results in non-trivial operations, creating a need for new tools.
A main example is that, in contrast with continuous value data, using dithering, i.e. adding
noise to a signal prior to its quantization, improves the estimation performance based on this signal \cite{Papadopoulos_Wornell_Oppenheim_2001,gustafsson2013generating,dabeer2008multivariate}. 
However, while there are various estimation algorithms based on quantized
signals, there has been less emphasis on the analysis and design of mixed-resolution architectures.


In this work, we consider Bayesian parameter estimation in systems with mixed-resolution, analog and 1-bit quantized, measurements. We develop the LMMSE estimator and its associated MSE for the considered model.
We present the linear Gaussian orthonormal (LGO) measurement model, which is shown to generalize common schemes, including: 1) scalar parameter estimation and 2) channel estimation in mixed-ADC massive MIMO communication systems. A closed-form analytic expression of the MSE of the LMMSE estimator is derived under the LGO measurement model. The resource allocation problem is formalized under the LGO model with the tractable expression of the MSE as the objective function and under power constraints, and solved using a one-dimensional search over value pairs. The concept of dithering, the addition of noise to the measurements before quantization, is presented for the mixed-resolution scheme and the resource allocation problem is solved while also optimizing the dithering noise.
Finally, simulations for the LGO model have been conducted and have shown the advantages of mixed-resolution estimation compared with purely-quantized or purely-analog settings, for the scalar case and for channel estimation in massive MIMO. In addition, the possible improvement from dithering can be seen even when adding the noise to both the analog and quantized measurements.

The remainder of the paper is organized as follows: Section \ref{II} presents the general mixed-resolution measurement model and the resource allocation problem which is shown to be computationally tedious. In Section \ref{III}, the LGO measurement model is presented and the resource allocation problem is solved for the LGO model, including dithering design. In Section \ref{IV}, special cases of the LGO model are discussed. Simulations of the resource allocation method are given in Section \ref{V}. Finally, our conclusions can be found in Section \ref{VI}.

\textit{Notation:} We use boldface lowercase letters to denote vectors and boldface capital letters for matrices. The identity matrix of size $M\times M$ is denoted by $\Imat_M$ and vector of ones of length $N$ is denoted by $\onevec_N$. The symbols $(\cdot)^*$,$(\cdot)^T$, and $(\cdot)^H$ represent the conjugate, transpose, and conjugate transpose operators, respectively. The symbol $\otimes$ is the Kronecker product.  We use $\text{trace}\left(\Amat\right)$ to denote the trace of the  matrix $\Amat$, and $\text{diag}\left(\Amat\right)$ to denote a diagonal matrix containing only the diagonal elements of $\Amat$. The $\text{arcsin}(\cdot)$ fucntion, when applied  to a vector or matrix, is applied elementwisely.
The distribution of a circularly symmetric complex Gaussian random vector with mean $\muvec$ and covariance matrix $\Sigmavec$ is denoted by $\mathcal{CN}(\muvec,\Sigmavec)$ and from here on noted as complex Gaussian.  We denote the covariance matrix of a vector $\avec$ as $\Cmat_\avec = {\rm{E}}[\avec \avec^H]$ and correlation between vectors $\avec$ and $\bvec$ as $\Cmat_{\avec\bvec}={\rm{E}}[\avec \bvec^H]$. The set of non-negative integers is denote by $\mathbb{Z}_+$. The 1-bit element-wise quantization function is applied separately on the real, $\text{Re}(z)$, and imaginary, $\text{Im}(z)$, part of any complex number $z\in\mathbb{C}$, and is defined as
\be \label{I.1}
	\mathcal{Q}(z) = \frac{1}{\sqrt{2}} \left[ \begin{cases} 
		\hspace{2.5mm} 1 \:, \text{Re}(z) \geq 0\\
		-1  \:, \text{Re}(z) < 0
	\end{cases} \hspace{-3mm}+ j\begin{cases} 
		\hspace{2.5mm} 1 \:, \text{Im}(z) \geq 0\\
		-1  \:, \text{Im}(z) < 0
	\end{cases}	\hspace{-3mm}\right].
\ee

\section{System Model} \label{II}
In this section we present the system model and introduce the problem of parameter estimation using mixed-resolution measurements. In Subsection \ref{II.A} the measurement model is presented and in Subsection \ref{II.B}, the LMMSE estimator and its associated MSE are derived for the discussed model. Finally, in Subsection \ref{II.C} we present the resource allocation optimization problem and discuss the difficulties of solving the problem.

\subsection{General Mixed-Resolution Measurement Model} \label{II.A}
We consider the problem of estimating a random parameter vector based on mixed-resolution data. In particular, we assume a parameter vector, $\thetavec\in\mathbb{C}^M$, with a zero-mean complex Gaussian distribution, $\thetavec\sim\mathcal{CN}(\zerovec,\Sigmavec_\thetavecsmall)$, where $\Sigmavec_\thetavecsmall$ is a known positive definite covariance matrix. The goal is to  estimate $\thetavec$ from a linear measurement model having both analog, high-resolution measurements:
\be \label{II.A.1}
	\xvec_a = \Hmat\thetavec + \wvec_a,
\ee
and quantized, low-resolution measurements:
\be	\label{II.A.2}
	\xvec_q = \mathcal{Q}\left(\Gmat\thetavec + \wvec_q \right),
\ee
where the quantization operator, $\mathcal{Q}(\cdot)$, is defined in \eqref{I.1}. The matrices $\Hmat\in\mathbb{C}^{N_a\times M}$ and $\Gmat\in\mathbb{C}^{N_q\times M}$ are known, with $N_a$ and $N_q$ being the number of analog and quantized measurements, respectively, and the added noise vectors, $\wvec_a$ and $\wvec_q$, are independent, zero-mean, complex Gaussian noise, i.e. $\wvec_a \sim \mathcal{CN}(\zerovec,\sigma_a^2 \Imat_{N_a})$ and $\wvec_q \sim \mathcal{CN}(\zerovec,\sigma_q^2 \Imat_{N_q})$. It is also assumed that the noise vectors, $\wvec_a$ and $\wvec_q$, and the unknown parameter vector, $\thetavec$, are mutually independent.
As a result, the analog measurements follow a complex Gaussian distribution, i.e. $\xvec_a \sim \mathcal{CN} \left(\zerovec, \Cmat_{\xvec_a}\right)$, with the covariance matrix
\be \label{II.A.3}
	\Cmat_{\xvec_a} = \Hmat\Sigmavec_\thetavecsmall\Hmat^H +\sigma_a^2 \Imat_{N_a}.
\ee
Similarly, the vector
\be	\label{II.A.4}
	\yvec \define \Gmat \thetavec + \wvec_q
\ee
is a complex Gaussian vector, i.e. $\yvec\sim \mathcal{CN} \left(\zerovec, \Cmat_\yvec\right)$, with the covariance matrix
\be \label{II.A.5}
	\Cmat_\yvec= \Gmat\Sigmavec_\thetavecsmall\Gmat^H+\sigma_q^2 \Imat_{N_q}.
\ee	
In particular, \eqref{II.A.2} implies that ${\rm{E}}[\xvec_q] = \zerovec$. However, the distribution of the quantized measurements $\xvec_q$ in \eqref{II.A.2} does not have a closed-form expression for the general case.
The goal is to use mixed-resolution measurements, i.e. the augmented vector
\be \label{II.A.6}
	\xvec \define \begin{bmatrix}
		\xvec_a^T & \xvec_q^T
	\end{bmatrix}^T,
\ee
to estimate $\thetavec$ efficiently. 

The considered model is fundamental in various signal processing applications with mixed-resolution data. An important case of this model, which is discussed in detail in Subsection \ref{IV.B}, is channel estimation in massive MIMO communication systems.
In this case, there are multiple users transmitting data to multiple antennas and the goal is to estimate the channel between users and each antenna, which is the unknown parameter vector $\thetavec$ in this case. Due to system limitation, such as sensor power consumption, sensor cost, and channel capacity, part of the observed data is quantized at the antenna and sent the fusion center for estimation.
This scenario, presented schematically in Fig. \ref{II.A.Fig1}, can be interpreted as a joint distributed-centralized estimation setup in a MIMO communication system with partially-quantized measurements.
\begin{figure}[htb]
	\centering
	\includegraphics[width=1\linewidth]{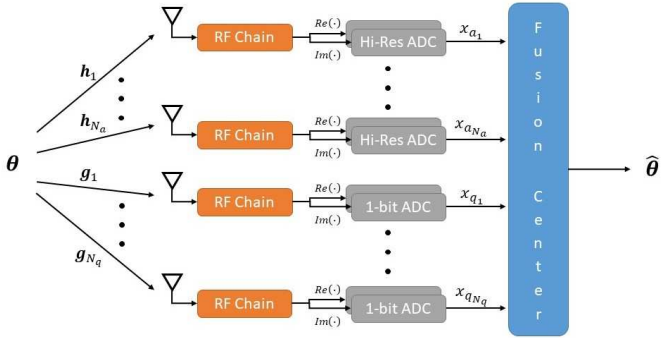}
	\caption{Schematic system model of estimation with mixed-resolution data. A data vector, $\thetavec$, is transmitted and received over a known channel and undergoes either a high- or 1-bit low-resolution quantization. The unknown data vector is then estimated from the mixed-resolution measurements in the fusion center.}
	\label{II.A.Fig1}
\end{figure}

\subsection{LMMSE Estimation} \label{II.B}
The MMSE estimator of $\thetavec$ based on the mixed-resolution data $\xvec$ in \eqref{II.A.6} is given by the conditional expectation, $\hat{\thetavec}^{\text{MMSE}}= {\rm{E}}[\thetavec\vert\xvec]$. Derivation of the MMSE estimator requires an analytic form of the conditional probability distribution function, $f(\thetavec\vert\xvec)$, which does not have a closed-form expression in general in the presence of quantized measurements. Moreover, since the MMSE estimator is a function of both the analog (continuous valued) measurements and the quantized (discrete valued) measurements, then, even the numerical evaluation of $\hat{\thetavec}^{\text{MMSE}}$ is intractable and requires multidimensional numerical integration.
Therefore, usually the LMMSE estimator is used when quantized measurements are involved.

For the sake of simplicity of presentation, in the following, $\hat{\thetavec}$ denotes the LMMSE estimator. For zero-mean measurements, as in our case, and under the assumption that $\Cmat_\xvec$ is a non-singular matrix, the LMMSE estimator, based on both $\xvec_a$ and $\xvec_q$, is given by
\be \label{II.B.1}
	\hat{\thetavec} = \Cmat_{\thetavecsmall \xvec}  \Cmat_{\xvec}^{-1} \: \xvec,
\ee
and the associated MSE of the LMMSE estimator is
\be	\label{II.B.2}
	\begin{aligned}
		MSE &= {\rm{E}}\left[(\hat{\thetavec} - \thetavec)^H(\hat{\thetavec} - \thetavec)\right] \\ &= \text{trace}\left(\Sigmavec_\thetavecsmall\right) - \text{trace}\left(\Cmat_{\thetavecsmall \xvec} \Cmat_{\xvec}^{-1} \Cmat_{\thetavecsmall \xvec}^H \right).
	\end{aligned}
\ee
In Appendix \ref{VII} it is shown that the auto-covariance matrix of $\xvec$ and the cross-covariance matrix of $\xvec$ and $\thetavec$ are block matrices given by
\be \label{II.B.3}
\begin{aligned}
	&\Cmat_\xvec = \left[\begin{matrix*}[l]
		\Hmat\Sigmavec_\thetavecsmall\Hmat^H+\sigma_a^2 \Imat_{N_a} & \vdots\\ \left(\sqrt{\frac{2}{\pi}} \Hmat \Sigmavec_\thetavecsmall \Gmat^H \left(\text{diag}\left(\Cmat_\yvec\right)\right)^{-\frac{1}{2}}\right)^H &\vdots
	\end{matrix*}\right. \\
	&\quad \quad \left. \begin{matrix*}[r]
		\sqrt{\frac{2}{\pi}} \Hmat \Sigmavec_\thetavecsmall \Gmat^H \left(\text{diag}\left(\Cmat_\yvec\right)\right)^{-\frac{1}{2}} \\
		\begin{aligned}
			&\frac{2}{\pi}\left(\text{arcsin}\left(\left(\text{diag}\left(\Cmat_\yvec\right)\right)^{-\frac{1}{2}} \text{Re}(\Cmat_\yvec) \left( \text{diag} \left( \Cmat_\yvec \right) \right)^{-\frac{1}{2}} \right) \right. \\
			&+\left. j\text{arcsin} \left(\left(\text{diag} \left(\Cmat_\yvec\right) \right)^{-\frac{1}{2}} \text{Im}(\Cmat_\yvec) \left( \text{diag} \left( \Cmat_\yvec \right) \right)^{-\frac{1}{2}} \right) \right)
		\end{aligned}
	\end{matrix*}\right]
\end{aligned}
\ee
and
\be \label{II.B.4}
\Cmat_{\thetavecsmall \xvec} = \begin{bmatrix}
	\Sigmavec_\thetavecsmall \Hmat^H & \sqrt{\frac{2}{\pi}} \Sigmavec_\thetavecsmall \Gmat^H \left(\text{diag}\left(\Cmat_\yvec\right)\right)^{-\frac{1}{2}}
\end{bmatrix},
\ee
respectively. By substituting \eqref{II.B.3} and \eqref{II.B.4} into \eqref{II.B.1} and \eqref{II.B.2} we obtain the LMMSE estimator and its associated MSE.

\subsection{Optimization of MSE} \label{II.C}
The main goal in this paper is to find the optimal number of analog and quantized measurements, $N_a^*$ and $N_q^*$, respectively, in the sense of minimum MSE of the associated LMMSE estimator, given in \eqref{II.B.2} under some physical constraints. That is, we aim to solve the following optimization problem,
\be \label{II.C.1}
	\begin{aligned}
		\min_{N_a,N_q} \quad &\text{trace}\left(\Sigmavec_\thetavecsmall\right) - \text{trace}\left(\Cmat_{\thetavecsmall \xvec} \Cmat_{\xvec}^{-1} \Cmat_{\thetavecsmall \xvec}^H \right) \\
		&\text{s.t.} \quad \begin{cases}
			\hvec_1(N_a,N_q)\leq\zerovec\\
			\hvec_2(N_a,N_q)=\zerovec \\
			N_a,N_q \in \mathbb{Z}_+
		\end{cases}
	\end{aligned},
\ee
where $\hvec_1(N_a,N_q)\leq0$ and $\hvec_2(N_a,N_q)=0$ represent different inequality and equality constraints, respectively, which stem from physical system requirements.
The optimization problem in \eqref{II.C.1} is an integer programming problem, which 
often  leads to solutions of combinatorial nature that cannot be solved in a reasonable time, even for small datasets \cite{Boyd_2004}.  Moreover, since the decision variables $N_a$ and $N_q$ represent the dimensions of the matrices $\Cmat_{\thetavecsmall \xvec}$ and $\Cmat_\xvec$ in the objective function of \eqref{II.C.1}, 
the problem cannot be solved by a simple relaxation that allows non-integer rational solutions.
As a result, for each value pair of the decision variables, $N_a$ and $N_q$, we need to calculate the inverse matrix, $\Cmat_\xvec^{-1}$, and perform matrix multiplication, giving a total computational complexity of $\mathcal{O}\left((N_a+N_q)^3+2M(N_a+N_q)^2\right)$, which increases as the number of measurements increases.
This approach may be intractable and may hinder insights into the original problem. 

\section{Optimization for the Orthonormal Measurement Model} \label{III}
In this section we present the LGO measurement model, optimize the resource allocation for this model, and propose the use of dithering.
In Subsection \ref{III.A} we present the LGO model and derive a tractable expression for the MSE under this model. In Subsection \ref{III.B} we discuss some physical constraints common in real-world systems with mixed-resolution measurements and in Subsection \ref{III.C} solve the resource allocation optimization problem under a constraint. In Subsection \ref{III.D} the concept of dithering for the LGO model is discussed and two possible cases of the optimization problem are presented while allowing dithering.

\subsection{Orthonomal Measurement Model} \label{III.A}
The LGO model is the model described in Subsection II.A, which also satisfies the following assumptions:
\begin{enumerate}[{A}.1)]
	\item The elements of $\thetavec$ are uncorrelated with unit variance, i.e. $\Sigmavec_\thetavecsmall = \Imat_M$.
	\item The matrix $\Hmat$ is a block matrix of size $N_a\times M$, where $N_a = M n_a$: 
	\be \label{III.A.1}
		\Hmat = \begin{bmatrix}
		\Hmat_1 & \Hmat_2 & \dots & \Hmat_{n_a}
		\end{bmatrix}^T,
	\ee
	where each block satisfies
	\be \label{III.A.2}
		\Hmat_i^H \Hmat_j = \rho_a\Imat_M \:, \quad i=j
	\ee
	and $\rho_a>0$. If $i \neq j$ then the product $\Hmat_i^H \Hmat_j$ can take arbitrary values.
	\item The matrix $\Gmat$ is a block matrix of size $N_q\times M$, where $N_q = M n_q$, with equal blocks:
	\be \label{III.A.3}
		\Gmat = \onevec_{n_q} \otimes \Gmat_1 ,
	\ee
	where
	\be \label{III.A.4}
		\Gmat_1^H \Gmat_1 = \rho_q\Imat_M \; ,
	\ee
	in which $\rho_q>0$.
\end{enumerate}
\begin{theorem} \label{Theorem}
	The LMMSE estimator for the mixed-resolution model, described in Subsection \ref{II.A} and under Assumptions A.1-A.3,  is
	\be \label{III.A.5}
		\begin{aligned}
			&\hat{\thetavec} = \left[\begin{matrix*}[l]
				\left(\frac{1}{\rho_a n_a + \sigma_a^2} - \frac{2\rho_q n_q\sigma_a^2}{\pi (\rho_q+\sigma_q^2) (\alpha + \beta(n_a) \rho_q n_q)(\rho_a n_a + \sigma_a^2)^2} \right) \Hmat^H & \vdots \end{matrix*}\right. \\
			& \hspace{1.5cm} \left. \begin{matrix*}[r]
			\sqrt{\frac{2}{\pi (\rho_q+\sigma_q^2)}} \frac{\sigma_a^2}{(\alpha + \beta(n_a) \rho_q n_q)(\rho_a n_a + \sigma_a^2)} \Gmat^H
		\end{matrix*}\right] \xvec 
		\end{aligned}
	\ee
	and its associated MSE is
	\be	\label{III.A.6}
		\begin{aligned}
			&MSE = M - M\left(\frac{\rho_a n_a }{\rho_a n_a + \sigma_a^2} \right.\\
			&\hspace{0.6cm} \left. +\frac{2\rho_q n_q \sigma_a^4}{\pi (\rho_q+\sigma_q^2) \left(\alpha + \beta(n_a)\rho_q n_q\right) \left(\rho_a n_a + \sigma_a^2\right)^2}\right)
		\end{aligned},
	\ee
	where 
	\be	\label{III.A.7}
		\alpha \define 1-\frac{2}{\pi}\arcsin\left(\frac{\rho_q}{\rho_q+\sigma_q^2} \right) = \frac{2}{\pi} \arccos\left(\frac{\rho_q}{\rho_q+\sigma_q^2} \right) 
	\ee
	and
	\be	\label{III.A.8}
		\beta(n_a) \define \frac{2}{\pi} \arcsin\left(\frac{\rho_q}{\rho_q+\sigma_q^2} \right) \frac{1}{\rho_q} - \frac{2 \rho_a n_a}{\pi(\rho_q +\sigma_q^2)(\rho_a n_a + \sigma_a^2)}.
	\ee
\end{theorem}

\begin{proof}
	The proof appears in Appendix \ref{VIII}.
\end{proof}
The MSE in \eqref{III.A.6} does not require matrix inversion and is only a function of the scalar variables: $n_a$, $n_q$, $\rho_a$, $\rho_q$, $\sigma_a^2$, and $sigma_q^2$. Thus, it can be used to solve various optimization problems.
The two extreme cases of Theorem 1 are when $n_a=0$ and when $n_q=0$. For these cases, substituting $n_a=0$ in \eqref{III.A.6}, we obtain
\be \label{III.A.9}
	MSE = M - \frac{2M\rho_q n_q}{\pi (\rho_q+\sigma_q^2) \left(\alpha + (1-\alpha)n_q\right)}.
\ee
Similarly, substituting $n_q=0$ in \eqref{III.A.6}, we obtain
\be \label{III.A.10}
	MSE = M - \frac{M\rho_a n_a}{\rho_a n_a + \sigma_a^2}.
\ee
In both cases, the MSE given in \eqref{III.A.9} and \eqref{III.A.10} is a monotonically decreasing function of $n_q$ and $n_a$, respectively.
The expression in \eqref{III.A.9} coincides with the result in \cite{li2017channel} for purely quantized data.

\subsection{Constraints} \label{III.B}
Physical distributed networks, such as sensor networks, typically suffer from energy constraints and limited communication bandwidth, requiring quantization before data can be transmitted to a fusion center for further processing.
In this context, one of the incentives of integrating low-resolution ADCs is their power consumption, which is much lower than that of high-resolution ADCs. In particular, power consumed in each ADC can be expressed as a factor of the number of quantization bits, $\tilde{b}$, as follows:
\be \label{III.B.1}
	P_{ADC} = \text{FOM}_\text{W} f_s 2^{\tilde{b}},
\ee
where $f_s$ is the sampling rate and $\text{FOM}_\text{W}$ is Walden's figure-of-merit for evaluating the power efficiency with ADCs resolution and speed \cite{walden1999analog}.
In this paper, the power consumption of a single high-resolution ADC with $b$ bits is denoted by $P_H$ and that of a single 1-bit ADC by $P_L$. Thus, the total power consumption of the $N_a$ high-resolution measurements is $N_a P_H$ and for the $N_q$ low-resolution measurements is $N_q P_L$. Due to power limits of physical systems, we consider that the following constraint is imposed on the total power:
\be \label{III.B.2}
	N_a P_H + N_q P_L \leq P_{\text{max}}.
\ee
Substituting the power consumption of each ADC from \eqref{III.B.1} with $\tilde{b}=b$ and $\tilde{b}=1$ for the $b$-bit and 1-bit measurements, respectively, into \eqref{III.B.2}, we obtain the constraint
\be \label{III.B.3}
	2^b N_a + 2 N_q \leq \tilde{P}_{\text{max}},
\ee
where $\tilde{P}_{\text{max}}\define\nicefrac{P_{\text{max}}}{\text{FOM}_\text{W} f_s}$ is the normalized maximum power. By substituting $N_a=Mn_a$ and $N_q=Mn_q$, from assumption A.2 and A.3 of the LGO model, \eqref{III.B.3} can be rewritten as
\be \label{III.B.4}
	2^bMn_a + 2Mn_q \leq \tilde{P}_{\text{max}}.
\ee
It should be noted that while the constraint in \eqref{III.B.2} treats the high-resolution measurements as finite $b$-bit quantized data, the MSE in \eqref{III.A.6} is derived under the assumption of pure analog measurements. Therefore, the number of bits that are used to represent the high-resolution ADC, $b$, should be chosen such that the quantization error is negligible. In simulations, we demonstrate that by choosing $b$ large enough, the approximation of analog measurements holds and the MSE from \eqref{III.A.6} is achieved by $b$-bit quantized data.

In addition to the power constraints in \eqref{III.B.3}, systems may also have other physical constraints on the number of measurements. This can be due to system design or available workspace, such as a field in which sensors are deployed, requiring a minimal distance from each other to avoid interference. The optimization in the following section can be readily extended to incorporate such constraints.

\subsection{Resource Allocation} \label{III.C}
In this subsection, we optimize the resource allocation of the LGO model using the analytical expression of the MSE derived in Subsection \ref{III.A} as an objective function and imposing the constraints from Subsection \ref{III.B}. We show that the integer programming problem from \eqref{II.C.1} can be solved in polynomial time.
It should be noted that the expression of the MSE in \eqref{III.A.6} can be used as a tractable objective function for different optimization problems of the LGO mixed-resolution scheme.

Under assumptions A.1-A.3, the objective function, i.e. the MSE from \eqref{II.B.2}, is now given by \eqref{III.A.6}. In addition, since $M$ is known, the decision variables can be changed to be $n_a$ and $n_q$, using the relation $N_a=M n_a$ and $N_q=M n_q$. The minimization of the MSE in the following is conducted under the power constraint in Subsection \ref{III.B}.
Thus, applying the constraint from \eqref{III.B.4} and substituting \eqref{III.A.6} in \eqref{II.C.1}, the minimum MSE problem under a power constraint is formulated as
\be	\label{III.C.1}
	\begin{aligned}
		\min_{n_a,n_q} \quad &M - M\left( \frac{\rho_a n_a}{\rho_a n_a + \sigma_a^2} \right. \\
		&\left.+ \frac{2\rho_q n_q \sigma_a^4}{\pi (\rho_q+\sigma_q^2) (\alpha + \beta(n_a)\rho_q n_q)(\rho_a n_a + \sigma_a^2)^2}\right)\\
		&\text{s.t.} \quad \begin{cases}
			2^bMn_a + 2M n_q \leq \tilde{P}_{\text{max}}\\
			n_a,n_q \in \mathbb{Z}_+
		\end{cases}.
	\end{aligned}
\ee
Solving the optimization problem in \eqref{III.C.1} no longer requires the inversion of $\Cmat_\xvec$. Moreover, the values $n_a$ and $n_q$ are no longer found in the matrix dimensions.
Thus, \eqref{III.C.1} can be solved using a standard  search approach or by a conventional relaxation approach. In this paper, we adopt the first option.
\begin{proposition} \label{Proposition}
	The optimization problem in \eqref{III.C.1} can be solved using a one-dimensional search over $n_a$, taking a set of discrete values
	\be \label{III.C.2}
		n_a \in \left\{0,1,\cdots,\left\lfloor \frac{\tilde{P}_{\text{max}}}{2^bM} \right\rfloor \right\}
	\ee
	where for each value of $n_a$, the value of $n_q$ is choosen to utilize maximum power, i.e.
	\be \label{III.C.3}
		n_q = \left\lfloor \frac{\tilde{P}_{\text{max}} - 2^bMn_a}{2M} \right\rfloor.
	\ee
\end{proposition}
\begin{proof}
	The proof appears in Appendix \ref{IX}.
\end{proof}
Based on Proposition \ref{Proposition}, one can numerically evaluate the optimal value, $n_a^*$, by using a simple one-dimensional search algorithm over the discrete values of $n_a$ described in \eqref{III.C.2} to minimize the MSE. Then, the optimal number of quantized measurements, $n_q^*$, is obtained by substituting $n_a=n_a^*$ in \eqref{III.C.3}. 
The optimal resource allocation depends on the system parameters: the total energy budget $\tilde{P}_{max}$ and the noise variances, $\sigma_a^2$ and $\sigma_q^2$.

In the following we present a few special cases to interpret the MSE in \eqref{III.A.6} and the resource allocation optimization in \eqref{III.C.1}.
\begin{itemize}	
	\item For the trivial case where the noise of the analog measurements approaches zero, i.e. $\sigma_a^2\to 0$, it can be seen that the MSE in \eqref{III.A.6} approaches zero as well for any $n_a\geq 1$. Therefore, in this case an optimal solution of \eqref{III.C.1} is obtained for $n_a^*=1$, which enables the estimation of $\thetavec$ without an estimation error such that the MSE is no longer a function of $n_q$. That is, additional measurements, both analog and quantized, won't change the optimal MSE value. 
	\item For the case where the noise of the quantized measurements approaches zero, i.e. $\sigma_q^2 \to 0$, the parameter $\alpha$, defined in \eqref{III.A.7}, also approaches zero.
	By substituting $\sigma_q^2 \to 0$ and $\alpha\to 0$ in \eqref{III.A.6}, we obtain that the MSE in this case is given by
	\be	\label{III.C.4}
		\begin{aligned}
			\lim_{\sigma_q\rightarrow 0} &MSE=M - M\left(\frac{\rho_a n_a}{\rho_a n_a + \sigma_a^2} \right.\\&+ \left.\frac{2 \sigma_a^4}{\pi (\rho_a n_a +\sigma_a^2)^2 - 2\rho_a n_a (\rho_a n_a + \sigma_a^2)}\right) ,
		\end{aligned}
	\ee
	for $n_q\geq 1$ and by \eqref{III.A.10} for $n_q=0$.
	Depending on the available power, $\tilde{P}_{\text{max}}$, the optimal resource allocation scheme in this case is determined. It can be seen that for $n_q\geq 1$, the MSE given in \eqref{III.C.4} is a constant function w.r.t. $n_q$. That is, for the noiseless case, the number of quantized measurements does not change the MSE as long as there is at least a single quantized measurement. Therefore, for this case of $\sigma_q\rightarrow 0$, the optimal policy in the sense of minimum MSE of the LMMSE estimator is one of the two following options: 1) take the maximum possible number of analog measurements and at least a single quantized measurement or 2) use only analog measurements, i.e. $n_{a_{max}}\define\left\lfloor \frac{\tilde{P}_{\text{max}}}{2^bM} \right\rfloor$, as given in \eqref{III.A.10}.
	The choice between these two options is as follows. 
	If the available power maintains the following inequality
	\be \label{III.27}
		\tilde{P}_{\text{max}} - \left\lfloor \frac{\tilde{P}_{\text{max}}}{2^bM} \right\rfloor 2^bM \geq 2M
	\ee
	then Option 1 is the optimal solution of \eqref{III.C.1} for this case. If the inequality in \eqref{III.27} does not hold, the optimal solution is either Option 1 with $n_a^*=n_{a_{max}}-1$ and $n_q^*\geq 1$ or Option 2 with $n_a^*=n_{a_{max}}$  and $n_q^*=0$. 
	Comparing the MSE in \eqref{III.A.10} and \eqref{III.C.4} for the latter case,
	Option 2 is optimal if the following inequality holds
	\be \label{III.C.5}
		\left( (\pi-2) \rho_a^2 - 2\rho_a \sigma_a^2\right) n_{a_{max}}\hspace{-0.1cm} > \hspace{-0.05cm} 2\sigma_a^4 -\pi\rho_a\sigma_a^2 + (\pi-2)\rho_a^2
	\ee
	and if not, then Option 1 is optimal.
	Thus, the sampling policy in this case depends on the noise variance of the analog measurement, $\sigma_a^2$, the factor $\rho_a$, and the maximum number of available analog measurements under the power constraint from \eqref{III.B.2}, i.e. by the condition in \eqref{III.27}.	
\end{itemize}

\subsection{Optimization with Dithering} \label{III.D}
Dither, roughly speaking, is a random noise process added to a signal prior to its quantization \cite{Papadopoulos_Wornell_Oppenheim_2001,gustafsson2013generating,dabeer2008multivariate}.
The addition of dithering noise is commonly used in both Bayesian and non-Bayesian estimation with low-resolution quantized data. Although noise commonly degrades the performance of a system, it has been shown that the addition of noise to quantized measurements can improve system performance.

In this subsection, we consider the addition of independent, zero-mean, complex Gaussian dithering noise vectors, $\wvec_{d_a}\sim\mathcal{CN}(\zerovec,\sigma_{d_a}^2 \Imat_{N_a})$ and $\wvec_{d_q}\sim\mathcal{CN}(\zerovec,\sigma_{d_q}^2 \Imat_{N_q})$, to the analog and quantized measurements, respectively, before quantization. 
The dithering noise vectors, $\wvec_{d_a}$ and $\wvec_{d_q}$, and the vectors $\thetavec$, $\wvec_a$, and $\wvec_q$ are assumed to be mutually independent. Therefore, the analog measurement vector in \eqref{II.A.1} now equals
\be \label{III.D.1}
	\xvec_a = \Hmat\thetavec + \wvec_a + \wvec_{d_a},
\ee
and the quantized measurement vector in \eqref{II.A.2} satisfies
\be	\label{III.D.2}
	\xvec_q = \mathcal{Q}\left(\Gmat\thetavec + \wvec_q + \wvec_{d_q} \right).
\ee
Since the vectors $\thetavec$, $\wvec_{d_a}$, $\wvec_{d_q}$, $\wvec_a$, and $\wvec_q$ are all mutually independent, then it can be shown, similar to the derivation of \eqref{II.A.3} and \eqref{II.A.5}, that \eqref{III.D.1} and \eqref{III.D.2} imply, in this case, the following covariance matrices of the measurements: 
\be \label{III.D.3}
	\Cmat_{\xvec_a} = \Hmat \Sigmavec_\thetavecsmall \Hmat^H + (\sigma_a^2 + \sigma_{d_a}^2) \Imat_{N_a}
\ee
and
\be \label{III.D.4}
	\Cmat_{\yvec} = \Gmat \Sigmavec_\thetavecsmall \Gmat^H + (\sigma_q^2 + \sigma_{d_q}^2) \Imat_{N_q}.
\ee
Thus, all of the results developed in Subsections \ref{III.A} and \ref{III.C} hold with the noise variance of the analog measurements, $\sigma_a^2$, increasing by $\sigma_{d_a}^2$, and that of the quantized measurements, $\sigma_q^2$, increasing by $\sigma_{d_q}^2$. In particular, the MSE for the LGO model from \eqref{III.A.6} in this case is given by
\be \label{III.D.5}
	\begin{aligned}
		&MSE_{\hat{\thetavecsmall}} = M - M\left(\frac{\rho_a n_a }{\rho_a n_a + \sigma_a^2 + \sigma_{d_a}^2} \right.\\
		&\hspace{-0.1cm}\left. +\frac{2\rho_q n_q (\sigma_a^2+\sigma_{d_a}^2)^2}{\pi (\rho_q+\sigma_q^2+\sigma_{d_q}^2) \left(\alpha_d + \beta_d(n_a)\rho_q n_q\right) \left(\rho_a n_a + \sigma_a^2+\sigma_{d_a}^2\right)^2}\right)
	\end{aligned},
\ee
where
\be	\label{III.D.6}
	\alpha_d \define \frac{2}{\pi} \arccos\left(\frac{\rho_q}{\rho_q+\sigma_q^2+\sigma_{d_q}^2} \right)
\ee
and
\be	\label{III.D.7}
	\begin{aligned}
		\beta_d(n_a) \define &\frac{2}{\pi} \arcsin\left(\frac{\rho_q}{\rho_q+\sigma_q^2+\sigma_{d_q}^2} \right)\frac{1}{\rho_q} \\ & - \frac{2 \rho_a n_a}{\pi(\rho_q +\sigma_q^2+\sigma_{d_q}^2)(\rho_a n_a + \sigma_a^2+\sigma_{d_a}^2)},
	\end{aligned}
\ee
are the equivalent of \eqref{III.A.7} and \eqref{III.A.8}, respectively, and replacing $\sigma_a^2$ and $\sigma_q^2$ with $\sigma_a^2+\sigma_{d_a}^2$ and $\sigma_q^2+\sigma_{d_q}^2$, respectively.

Under this model, the goal is to find the optimal resource allocation, the number of analog and quantized measurements, which minimizes the estimator's MSE while also optimizing the variance of the added dithering noise.
Therefore, we look at the LGO Model from Subection \ref{III.A} and solve the optimization problem given in \eqref{III.C.1} with the objective function now being the MSE from \eqref{III.D.5} and adding the dithering noise variances, $\sigma_{d_a}^2$ and $\sigma_{d_q}^2$, as additional decision variables. Mathematically, the optimization problem in \eqref{III.C.1} with the addition of dithering noise can be rewritten as
\be	\label{III.D.8}
	\begin{aligned}
		&\min_{n_a,n_q,\sigma_{d_a}^2,\sigma_{d_q}^2} \quad M - M\left(\frac{\rho_a n_a}{\rho_a n_a + \sigma_a^2+\sigma_{d_a}^2} \right. \\
		&\left.+ \frac{2\rho_q n_q (\sigma_a^2+\sigma_{d_a}^2)^2}{\pi (\rho_q+\sigma_q^2+\sigma_{d_q}^2) (\alpha_d + \beta_d(n_a)\rho_q n_q)(\rho_a n_a + \sigma_a^2+\sigma_{d_a}^2)^2}\right)\\
		&\quad \text{s.t.} \quad \begin{cases}
			2^bMn_a + 2M n_q \leq \tilde{P}_{max}\\
			n_a,n_q \in \mathbb{N}_0\\
			\sigma_{d_a}^2 \geq 0 \\
			\sigma_{d_q}^2 \geq 0
		\end{cases},
	\end{aligned}
\ee
where $\alpha_d$ and $\beta_d$ are given in \eqref{III.D.6} and \eqref{III.D.7}, respectively.

It can be shown that for any given set of values $n_a$, $n_q$, and $\sigma_{d_q}^2$ the optimal dither noise added to the analog measurements, $\sigma_{d_a}^2$, is zero. This solution is intuitive since for analog data, the MSE decreases as the signal-to-noise ratio (SNR) increases, and, thus, the addition of a dithering noise can only degrade estimation performance. Therefore, in the following, we discuss two scenarios of dithering: 1) the optimal solution, which is obtained by adding dithering noise to the quantized measurements only; and 2) adding dithering noise with the same variance to the entire system, both analog and quantized measurements. Thus, we add a single decision variable to the optimization problem in \eqref{III.D.8}, denoted as $\sigma_d^2$, where for Scenario 1 we set $\sigma_{d_a}^2=0$ and $\sigma_{d_q}^2=\sigma_d^2$, and for Scenario 2 we set $\sigma_{d_a}^2=\sigma_{d_q}^2=\sigma_d^2$.

Similar to the derivation of Proposition \ref{Proposition}, it can be proved that for each value of $n_a$ and given the dithering noise variances, $\sigma_{d_a}^2$ and $\sigma_{d_q}^2$, the number of quantized measurements, $n_q$, should be chosen to be the maximum allowed under the power constraint. Works that utilize dithering, such as \cite{vlachos2018dithered,liang2016mixed}, find the optimal dithering variance by using an exahustive search. Similarly, we find the optimal allocation with dithering, i.e. the solution of \eqref{III.D.8}, by utilizing a two-dimensional search over each value pair, $n_a$ and $n_q$, while utilizing maximum power and for each pair search possible values of dithering variance $\sigma_d^2$. This algorithm is summarized in Algorithm \ref{Algo}.
\begin{algorithm}[htb]
	\SetAlgoLined
	\KwIn{$M$, $b$, $\tilde{P}_{max}$, $\sigma_{d_{max}}^2$, $\sigma_{d_{res}}^2$}
	\KwOut{$n_a^*$, $n_q^*$, ${\sigma_d^2}^*$}
	Initialize $MSE_{opt} = M$\\
	\For{$n_a= 0 : \left\lfloor \frac{\tilde{P}_{\text{max}}}{2^bM} \right\rfloor$}{
		$n_q = \left\lfloor \frac{\tilde{P}_{\text{max}} - 2^bMn_a}{2M} \right\rfloor$\\
		\For{$\sigma_d^2= 0 : \sigma_{d_{res}}^2 : \sigma_{d_{max}}^2$}{
			Calculate $MSE_{temp}$ by substituting $n_a$, $n_q$, and $\sigma_d^2$ in \eqref{III.D.5}\\
			\If {$MSE_{temp} < MSE_{opt}$}{
				Update optimal values: $MSE_{opt}$, $n_a^*$, $n_q^*$, ${\sigma_d^2}^*$.
			}
		}
	}
	\caption{Resource Allocation with Dithering}
	\label{Algo}
\end{algorithm}

\section{Special Cases} \label{IV}
In this section we discuss two special cases of the LGO model, described in Subsection \ref{III.C}. In Subsection \ref{IV.A}, we present the problem of estimating a scalar parameter from noisy measurements, a model which is widely used in WSN, for example \cite{Ribeiro_Giannakis_2006,Papadopoulos_Wornell_Oppenheim_2001}.
In Subsection \ref{IV.B}, we discuss the allocation of analog and quantized measurements for channel estimation in massive MIMO communication systems \cite{li2017channel,pirzadeh2018spectral,lu2014overview,jacobsson2015one,pirzadeh2017spectral,jacobsson2017throughput}.

\subsection{Estimation of a Scalar Parameter} \label{IV.A}
In this subsection, we are interested in estimating a scalar unknown parameter, $\theta\in\mathbb{C}$, with a zero-mean complex Gaussian distribution, $\theta\sim \mathcal{CN}(0,1)$, based on mixed-resolution data.
Suppose, for example, a WSN with $N$ sensors, where $N_a$ of them transmit analog measurements and $N_q$ transmit 1-bit quantized measurements, where $N=N_a+N_q$, to a central unit for estimation. In the case, \eqref{II.A.1} and \eqref{II.A.2} are reduced to
\be \label{IV.A.1}
	\xvec_a = \onevec_{N_a} \theta + \wvec_a
\ee
and
\be \label{IV.A.2}
	\xvec_q = \mathcal{Q}\left(\onevec_{N_q} \theta + \wvec_q\right),
\ee
respectively. We assume that $\wvec_a$ and $\wvec_q$ are independent, zero-mean complex Gaussian noise distributed $\wvec_a\sim\mathcal{CN}(\zerovec,\sigma^2\Imat_{N_a})$ and $\wvec_q\sim\mathcal{CN}(\zerovec,\sigma^2\Imat_{N_q})$.
This scalar estimation problem satisfies the LGO model assumptions: First, it can be seen that the distribution of the unknown parameter, $\theta\sim \mathcal{CN}(0,1)$, satisfies A.1). Second, $\Hmat=\onevec_{N_a}$ and $\Gmat=\onevec_{N_q}$ can be treated as block matrices with each entry being the scalar $1$, which in turn satisfies assumption A.2) and A.3) with $\rho_a =1$ and $\rho_q =1$. Satisfying the assumptions allows us to use Theorem \ref{Theorem} in order to find the optimal measurement allocation scheme for the scalar case.
It should be noted that since the dimensions of the auto-covariance matrix $\Cmat_\xvec$ are affected by the number of measurements, $N_a$ and $N_q$, and not by the size of the unknown parameter vector $\thetavec$. Solving the optimization problem in \eqref{II.C.1} still requires the inversion of $\Cmat_\xvec$ at each value pair for the scalar case. Therefore, the proposed tractable formulation in \eqref{III.C.1} is also relevant and important for the scalar case.
 
\subsection{Channel Estimation in Massive MIMO} \label{IV.B}
In this subsection, we consider the special case of channel estimation using analog and 1-bit quantized measurement in massive MIMO networks. Massive MIMO has a high potential of enabling technology beyond fourth generation (5G) cellular systems due to its advantages in terms of spectral efficiency, energy efficiency, and the ability to use low-cost low-power hardware \cite{larsson2014massive,puglielli2015design}.
The following model of mixed-ADC massive MIMO has been used in \cite{li2017channel,pirzadeh2018spectral,lu2014overview,jacobsson2015one,pirzadeh2017spectral,jacobsson2017throughput} and is described in detail due to its importance and to clarify the relation to the considered LGO model. 

We study the uplink of a single-cell multi-user MIMO system consisting of $K$ single antenna users transmitting independent data symbols simultaneously to a base station (BS) equipped with $L$ antennas.
We consider a block-fading model with coherence bandwidth $W_c$ and coherence time $T_c$. In this model, each channel remains constant in a coherence interval of length $T=T_cW_c$ symbols and changes independently between intervals. The coherence interval can be divided into two parts: the first part  is used for channel estimation, referred to as the training phase, while the second part is for data transmission. During training, all $K$ users simultaneously transmit their pilot sequences of $K$ mutually orthogonal pilot symbols. Therefore, the received signal during the training phase is 
\be	\label{IV.B.1}
	\Rmat = \sqrt{\rho}\Amat \Phimat^T + \Wmat ,
\ee
where $\Amat\in\mathbb{C}^{L\times K}$ is the channel matrix, $\rho$ is the pilot transmission power, 
$\Phimat\in\mathbb{C}^{K\times K}$ is the pilot signal matrix transmitted from the $K$ users, and $\Wmat$ is independent, zero-mean, complex Gaussian noise with each element distributed $\mathcal{CN}(\zerovec,\sigma^2)$. We assume the channel vectors are i.i.d. and denote the $l\text{th}$ row of $\Amat$ as $\avec_l$, where $\avec_l$ has a zero-mean complex Gaussian distribution, i.e. $\avec_l\sim\mathcal{CN}(\zerovec,\Imat_K)$. 
The pilot sequences are drawn from the pilot signal matrix $\Phimat$, where $\Phimat^H\Phimat=\Imat_K$.

Since we assume independent channels and noise, then the rows of $\Rmat$ from \eqref{IV.B.1} are mutually independent, enabling us to analyze each row separately. Therefore, the $l$th row of $\Rmat$ (viewed as a column vector), satisfies
\be \label{IV.B.2}
	\rvec_l = \sqrt{\rho}\Phimat\avec_l + \wvec_l ,~l=1,\ldots,L.
\ee
Due to the high power consumption of high-resolution ADCs and the less informative data of low-resolution ADCs, in many works both analog and quantized measurements are used to benefit from both worlds (see, e.g. \cite{zhang2016mixed,liang2016mixed} and references therein).
In order to achieve good performance there is a need for more measurements when working with quantized data as opposed to analog. Therefore, the pilot sequence, $\Phimat$, can be transmitted a number of times with the antennas switching between the high- and low-resolution ADCs at each transmit.

Let us transmit the $K$ pilot symbols $N$ times. For each channel $l\in\{1,\dots,L\}$, we denote the number of times the pilot sequence is transmitted with the analog ADC connected to the antenna as $n_{a_l}$ and with the 1-bit quantized ADC as $n_{q_l}$ where $n_{a_l}+n_{q_l}=N$. We can organize the measurements now such that the analog measurements from \eqref{IV.B.2} are
\be \label{IV.B.3}
	\rvec_{a_l} = \sqrt{\rho} \Phimat_a \avec_l + \wvec_a,
\ee 
where $\rvec_{a_l}\in\mathbb{C}^{n_{a_l}K}$ and $\Phimat_a$ is a $n_{a_l}K\times K$ block matrix defined as
\be	\label{IV.B.4}
	\Phimat_a = \onevec_{n_a} \otimes \Phimat.
\ee 
Similarly, the quantized measurements from \eqref{IV.B.2} are
\be	\label{IV.B.5}
	\rvec_{q_l} = \sqrt{\rho} \Phimat_q \avec_l + \wvec_q,
\ee
where $\rvec_{q_l}\in\mathbb{C}^{n_{a_q}K}$ and $\Phimat_q$ is a $n_{q_l}K\times K $ block matrix defined as
\be	\label{IV.B.6}
	\Phimat_q = \onevec_{n_q} \otimes \Phimat.
\ee
By setting $\avec_l=\thetavec$, $\rvec_{a_l}=\xvec_a$, and $\rvec_{q_l}=\xvec_q$ the measurement model in \eqref{IV.B.3} and \eqref{IV.B.5} coincides with the general measurement model in \eqref{II.A.1} and \eqref{II.A.2} where $M=K$. Moreover, we now show that Assumptions A.1-A.3 from Subsection \ref{III.A} are satisfied for the mixed-ADC massive MIMO model. First, the channel is modeled such that $\avec_l\sim\mathcal{CN}(\zerovec,\Imat_M)$ keeping the assumption that $\Sigmavec_\thetavecsmall=\Imat_M$. Thus, Assumption A.1) is satisfied. Second, by using \eqref{IV.B.3} and \eqref{IV.B.4}, it can be seen that the matrix $\Hmat=\sqrt{\rho}\Phimat_a$ is a block matrix of size $n_{a_l}K\times K$ and satisfies $(\sqrt{\rho}\Phimat)^H\sqrt{\rho}\Phimat=\rho\Imat_K$, thus Assumption A.2) is satisfied with $\rho_a = \rho$. Similarly, by using \eqref{IV.B.5} and \eqref{IV.B.6}, it can be seen that the matrix $\Gmat=\sqrt{\rho}\Phimat_q$ satisfies Assmuption A.3) with $\rho_q = \rho$.

The goal under this model is to estimate the channel of a system consisting of a single antenna, $L=1$, while using both a high- and low-resolution ADCs in the BS, which the antenna can switch between to acquire both analog and quantized measurements. The number of measurements taken, or equivalently the number of times the pilot signal is transmitted using each ADC resolution, is to be optimized to minimize the MSE of the LMMSE estimator, while not exceeding the power consumption at the BS.
Substituting $M=K$, $\rho_a=\rho_q=\rho$, and $\sigma_a^2=\sigma_q^2=\sigma^2$ in \eqref{III.C.1}, the resource allocation problem can be solved for the problem of channel estimation in massive MIMO systems. Similarly, the same substitution can be done in \eqref{III.D.8} allowing to find the optimal resource allocation with dithering.

This approach can be extended to the more general case, where there are $L\geq 1$ i.i.d. antennas or channels to estimate. The channel estimation problem for the channel \eqref{IV.B.1} can be decomposed into parallel estimation problems \cite{zeitler2012bayesian} solved separately for each
channel with the maximum power available equaling for example, to $P_{max}/L$, allocating each channel an equal power supply, thus allowing us to optimize the whole system while solving the problem for a single channel.

\section{Simulations} \label{V}
In this section we numerically evaluate the performance of the mixed-resolution system presented in Section \ref{III} and that of the proposed resource allocation optimization approach. In Subsection \ref{V.A} we simulate the scalar case from Subsection \ref{IV.A}, in Subsection \ref{V.A} we simulate the model of channel estimation in massive MIMO from Subsection \ref{IV.B}, and in Subsection \ref{V.C} we compare the run time of the proposed resource allocation approach and the brute-force approach in \eqref{II.C.1}.
As discussed in Subsection \ref{III.B}, in order for the quantization noise to be negligible we use $b=6$ bits on a quantization range $[-5,5]$ to represent our analog measurements. The noise added to the analog and quantized measurements is assumed to have the same variance, $\sigma_a^2 = \sigma_q^2 = \sigma^2$, and we set $\rho_a = \rho_q = 1$. Our results are averaged over 100 Monte-Carlo simulations.

\subsection{Scalar Parameter Estimation} \label{V.A}
In this subsection, estimation of a scalar parameter, $M=1$, as discussed in Subsection \ref{IV.A}, is evaluated.
The simulation results in Fig. \ref{V.Fig1} show the behavior of the MSE from \eqref{III.A.6} for the scalar case with different values of $n_a$ and $n_q$ as a function of the noise variance $\sigma^2$. It can be seen that the addition of measurements, be it analog or quantized, does not degrade the performance in terms of MSE. In addition, when only analog measurement are used, $n_q= 0$, the MSE monotonically decreases as $\sigma^2$ decreases. The same can not be said for the mixed-resolution case for which the behavior of the MSE is not even convex. For example, it can be seen that there are cases such as $n_a=1$ and $n_q=100$ that the addition of dithering noise to  both measurement types can improve the MSE which may be counterintuitive due to the behavior of pure analog measurement estimation.
\begin{figure}[htb]
	\centering
	\includegraphics[width=\linewidth]{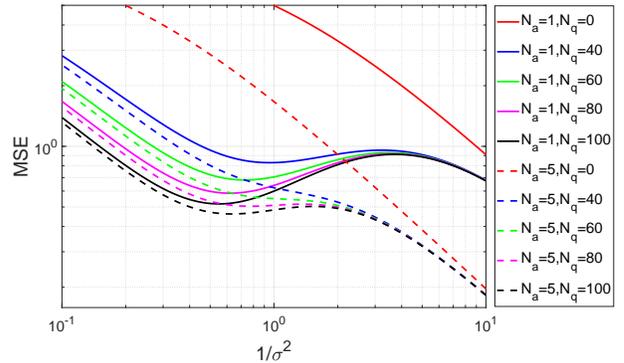}
	\caption{Scalar case: effects of noise variance on the estimator's MSE for different number of measurements.}
	\label{V.Fig1}
\end{figure}
In Fig. \ref{V.Fig2:Sig1} and Fig. \ref{V.Fig2:Sig2}, the MSE is shown for a given noise variance, $\sigma^2$, as a function of the number of analog measurements, $n_a$, and quantized measurements, $n_q$. Dots on the graphs show the possible value pairs of analog and quantized measurements, $n_a$ and $n_q$, for different values of available power, $\tilde{P}_{max}$. These figures show, as before, that taking more measurements does not increase the MSE. In addition, it can be seen that using a mixed-resolution approach can have a lower MSE than assuming a naive approach which utilizes only one type of measurement up to the maximum power available.
\begin{figure}[htb]
	\centering
	\subcaptionbox{\label{V.Fig2:Sig1}}[\linewidth]
	{ \includegraphics[width=.9\linewidth]{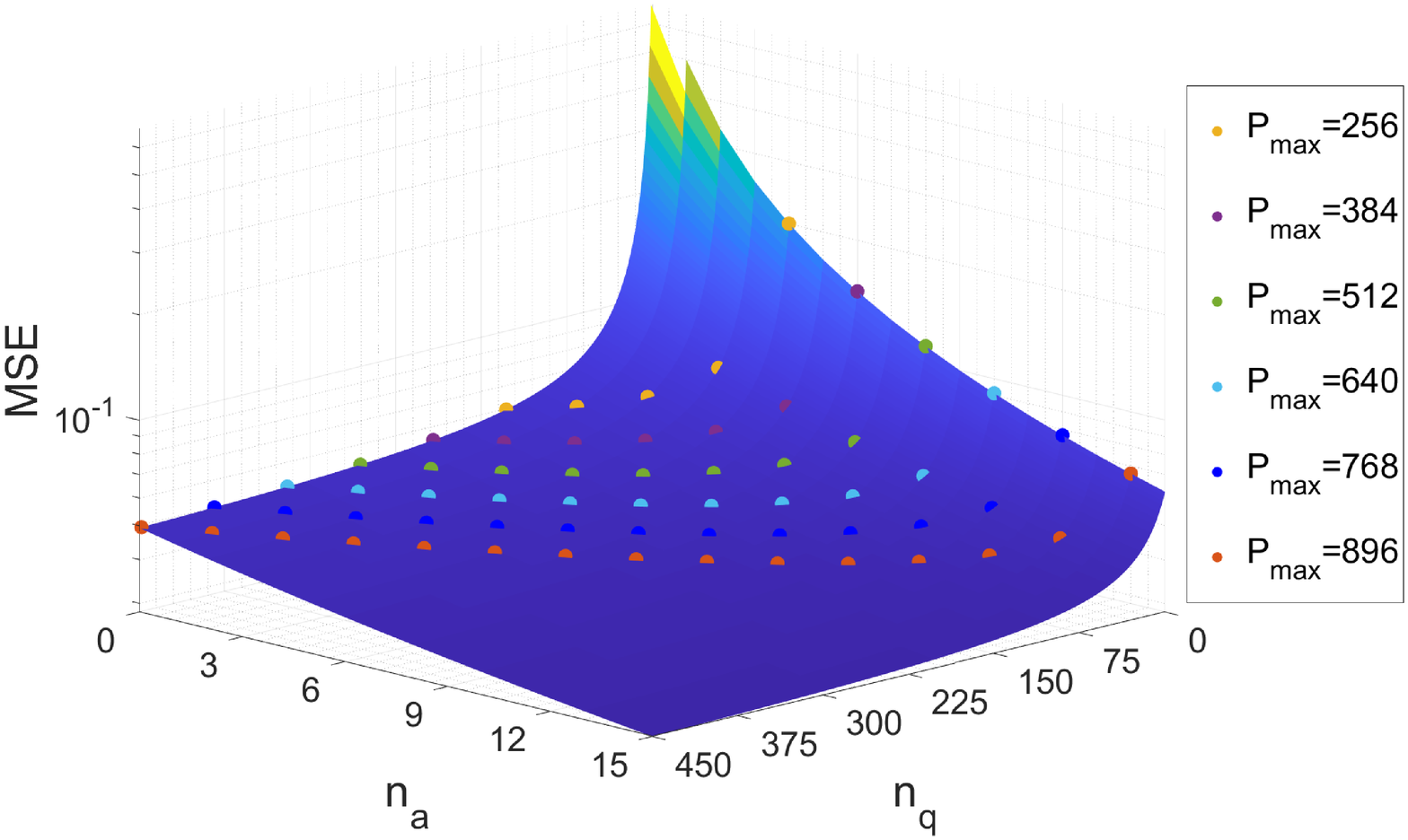}}
	\subcaptionbox{\label{V.Fig2:Sig2}}[\linewidth]
	{\includegraphics[width=.9\linewidth]{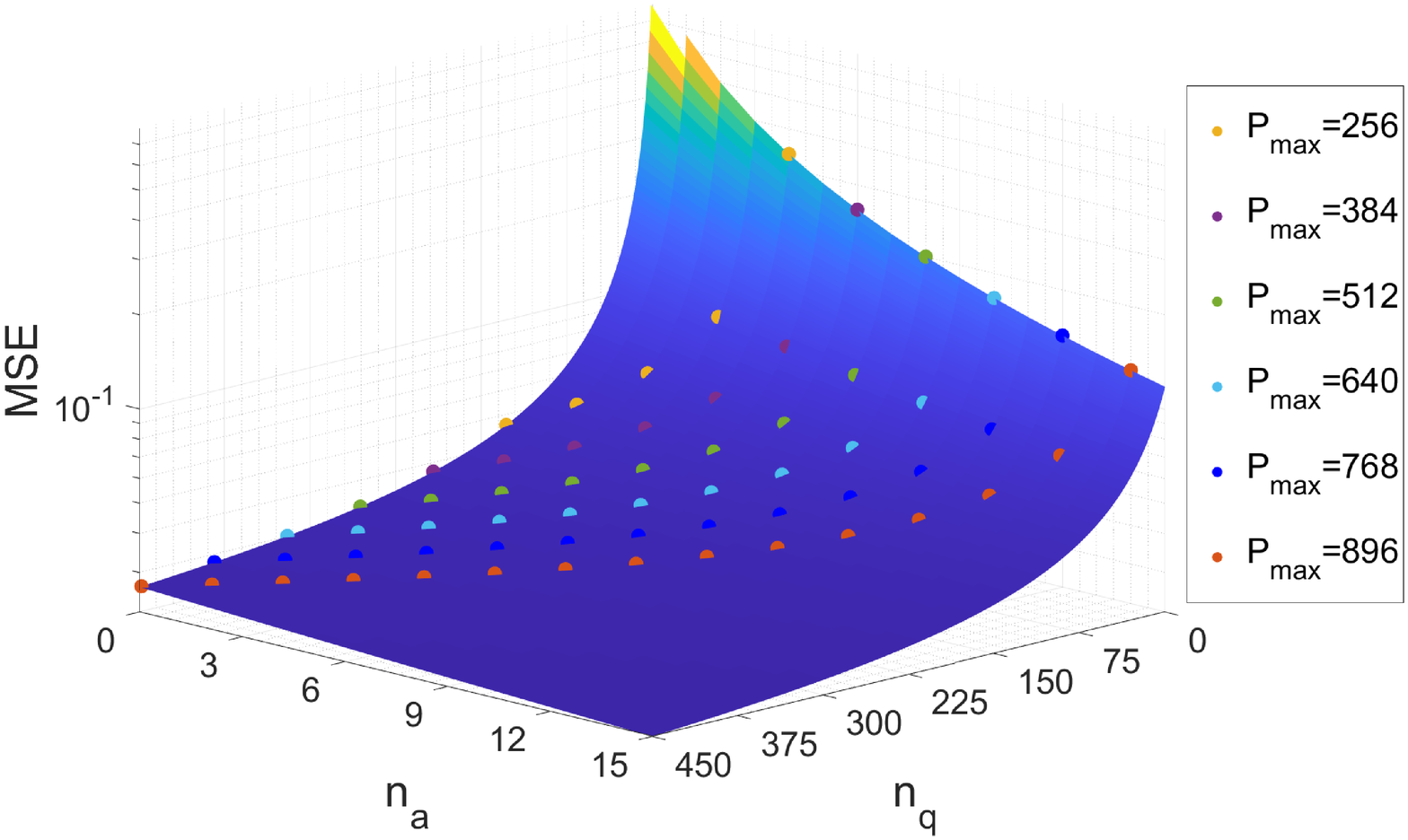}} 
	\caption{Scalar Case: The MSE as a function of the number of analog and quantized measurements, $n_a$ and $n_q$. The noise variance is $\sigma^2=1$ (a) and $\sigma^2=2$ (b), where $\sigma^2=\sigma_a^2=\sigma_q^2$.}
\end{figure}

\subsection{Channel Estimation in Massive MIMO} \label{V.B}
In this subsection, the massive MIMO model, as described in Subsection \ref{IV.B}, is simulated. The resource allocation optimization problem in \eqref{III.C.1} is solved for $M=10$ users transmitting a randomly generated pilot matrix $\Phimat$. The maximum power available is set to $\tilde{P}_{max} = 2^bMn_{a_{max}}$, where $b=6$ and $n_{a_{max}}=20$. In Fig. \ref{V.Fig3} we compare between the optimal resource allocation calculated using a one-dimensional search, as described in Subsection \ref{III.C}, and between two naive solutions: 1) a greedy scheme, which utilizes the maximum number of analog measurements and 2) a cost-effective scheme, which uses the maximum number of measurements by only having quantized measurements. The analytic MSE is calculated under the assumption of pure analog measurements, i.e. as given in \eqref{III.C.1}, although we use in practice $b=6$-level quantized data. Therefore, we also present Monte-Carlo simulations of the obtained MSE in practice that show that the values to represent the analog measurements make the quantization noise negligible for the purpose of estimation.

We can divide the graph into three sections: 1) low noise variance, $\sigma^2 <0.2$, in which case the use of all analog measurements is optimal, 2) high noise variance, $\sigma^2 > 2$,
 in which the use of all quantized measurements is optimal, and 3) middle section, $0.2<\sigma^2 <2$, in which it is optimal to use a mixed-resolution measurement scheme.
Moreover, Fig. \ref{V.Fig3} also compares the aforementioned solutions to the solution of the resource allocation optimization problem with optimization of the dithering noise added only to the quantized measurements. This is done using a two-dimensional search over $n_a$ and $\sigma_{q_d}^2$ with $\sigma_{q_d}^2\in[0,2]$ with increments of $0.1$, as in Algorithm \ref{Algo}. In the middle section in which better performance was achieved by using the mixed-resolution scheme as opposed to the naive solutions, the addition of dithering noise improved performance even more. As expected, in the low noise variance section, in which the analog measurements are optimal, the dithering had no effect but it did effect part of the high noise variance section improving the performance when using all quantized measurement. Therefore, we can conclude that utilizing dithering can improve system performance in the sense of the LMMSE estimator's MSE.
\begin{figure}[htb]
	\centering
	\hspace*{-0.3cm}
	\includegraphics[width=1.15\linewidth]{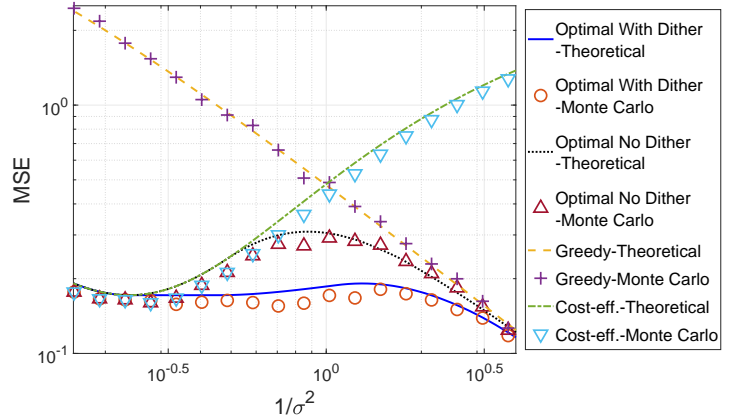}
	\caption{Channel estimation in massive MIMO with $M=10$ users. The MSE of the LMMSE estimator versus the noise variance, where $\sigma^2=\sigma_a^2=\sigma_q^2$. The addition of dithering noise is available only to the 1-bit quantized meausrements in the mixed-resolution scheme.}
	\label{V.Fig3}
\end{figure}

\subsection{Run Time} \label{V.C}
In this subsection, the computational complexity of solving the resource allocation optimization problem is evaluated under the LGO measurement model.
This is done by comparing the time of solving the optimization problem in \eqref{II.C.1} using the closed-form analytic expression of the MSE derived in Theorem \ref{Theorem} compared to the general MSE term in \eqref{II.B.2}, which requires matrix inversion. In both cases, the one-dimensional search from Proposition \ref{Proposition} is used.
It should be noted that the general MSE term in \eqref{II.B.2} does not give insight on the behavior of the MSE. Therefore, Proposition \ref{Proposition} isn't proven for the general term which in turn requires a two-dimensional search over all possible value pairs of $n_a$ and $n_q$ in order to solve the resource allocation optimization problem. Thus, for the general case, the run time of the optimization problem using the general MSE term is much longer than that presented in the simulation. 
The average computation time, “run time”, was evaluated by running the algorithm using Matlab on an Intel Xeon E5-2660 CPU.
In Fig. \ref{V.Fig4} we show the run time of the optimal resource allocation in both forms of the MSE as a function of the maximum number of analog measurements available. The maximum power available is set to equal $\tilde{P}_{max} = 2^b M n_{a_{max}}$ and we evaluate the cases of $M=1,3,10$.
It can be seen that the more analog measurements are available, the larger the maximum power and therefore, the computation time increases since the search is over more values. This has a larger effect when using matrix inversion since $\Cmat_\xvec\in \mathbb{C}^{(n_a+n_q)M \times (n_a+n_q)M}$. In addition, the size of the unknown parameter vector, $\theta \in \mathbb{C}^M$, does not affect the run time of the proposed approach in Theorem \ref{Theorem} while increasing the calculation time of the matrix inversion MSE in the direct approach.
\begin{figure}[htb]
	\centering
	\includegraphics[width=1\linewidth]{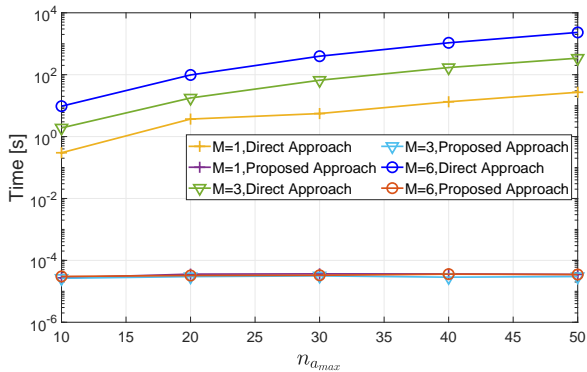}
	\caption{Run time comparison between calculation of the MSE using the direct approach and the closed-form expression in Theorem \ref{Theorem} for different values of $M$.}
	\label{V.Fig4}
\end{figure}

\section{Conclusion} \label{VI}
In this paper, we consider Bayesian parameter estimation using mixed-resolution measurements. First, we derive the LMMSE estimator and its associated MSE for the mixed-resolution case.
It is shown that the MSE requires matrix inversion with the size of the matrix depending on the number of analog and quantized measurements, and thus, optimization problems that aim to minimize the MSE w.r.t. the number of  measurements are impractical due to the exhaustive search required.
Next, we present the LGO model for which we calculate a closed-form expression for the LMMSE estimator and its corresponding MSE. Two special cases for which the mixed-resolution scheme under the LGO model is relevant are presented: 1) scalar parameter estimation which is used, for example, in WSN, and 2) channel estimation in massive MIMO communication systems.
Based on the closed-form expression of the MSE and enforcing a power consumption constraint, a resource allocation optimization problem is formulated with the goal of finding the optimal resources, namely the number of analog and quantized measurements. A one-dimensional search is proven to be sufficient in finding the optimal solution to the problem.
Furthermore, the concept of dithering is presented and the resource allocation optimization problem is derived while also allowing optimization of the dithering noise variance added to the system. 
Finally, in the simulations we show that the mixed-resolution scheme outperforms the naive approaches of pure analog or pure quantized measurements for certain ranges of noise variance that are not in the asymptotic region nor the so-called non-informative region. In addition, the possible benefits of dithering on estimation performance in terms of MSE are shown for mixed-resolution schemes.
Solving the resource allocation optimization problem for the LGO measurement model has a low-complexity solution allowing fast calculation. The mixed-resolution scheme can and should be adopted in different real-world applications with the solution for the LGO measurement model easily achieved and thus improve system performance.

\appendices

\section{Derivation of \eqref{II.B.3} and \eqref{II.B.4}} \label{VII}
In this appendix we develop the auto-covariance and cross-covariance matrices from \eqref{II.B.3} and \eqref{II.B.4} for the general model described in Subsection \ref{II.A}. By using \eqref{II.A.6}, it can be verified that the auto-covariance and cross-covariance matrices are block matrices given by
\be	\label{VII.1}
	\Cmat_\xvec = \begin{bmatrix}
		\Cmat_{\xvec_a} & \Cmat_{\xvec_a \xvec_q} \\ \Cmat_{\xvec_q \xvec_a} & \Cmat_{\xvec_q}
	\end{bmatrix}
\ee
and
\be \label{VII.2}
\Cmat_{\thetavecsmall \xvec} = \begin{bmatrix}
	\Cmat_{\thetavecsmall \xvec_a} & \Cmat_{\thetavecsmall \xvec_q}
\end{bmatrix},	
\ee
respectively, where $\Cmat_{\xvec_a}$ is given in \eqref{II.A.3}. From the analog measurement vector $\xvec_a$ given in \eqref{II.A.1} and based on the measurement model in Subsection \ref{II.A}, it can be verified that
\be \label{VII.3}
	\Cmat_{\thetavecsmall \xvec_a} = \Sigmavec_\thetavecsmall \Hmat^H.
\ee
To calculate the covariance matrix $\Cmat_{\xvec_q}$ we use the arcsine law (p. 396 in  \cite{PapoulisAthanasios2002Prva}) which implies that given two zero-mean jointly complex Gaussian random variables, $r$ and $t$, the cross-covariance of the quantized variables, $\mathcal{Q}(r)$ and $\mathcal{Q}(t)$, is given by
\be \label{VII.4}
	\begin{aligned}
		&\Cmat_{\mathcal{Q}(r),\mathcal{Q}(t)} = {\rm{E}}\left[ \mathcal{Q}(r) \mathcal{Q}^*(t) \right] \\
		&\quad = \frac{2}{\pi} \left[\text{arcsin}\left(\frac{\text{Re}\left\{\Cmat_{r t}\right\}}{\sqrt{\sigma_t^2 \sigma_r^2}} \right) + j\text{arcsin}\left(\frac{\text{Im}\left\{\Cmat_{r t}\right\}}{\sqrt{\sigma_t^2 \sigma_r^2}} \right) \right],
	\end{aligned}
\ee
where $\sigma_t^2$ and $\sigma_r^2$ are the covariance of the random variables $t$ and $r$, respectively. Therefore, given the measurement vector $\xvec_q$ in \eqref{II.A.2}, which is the 1-bit quantization of $\yvec$ in \eqref{II.A.4}, and applying the result in \eqref{VII.4} element-wise, the auto-covariance matrix is given by
\be \label{VII.5}
	\begin{aligned}
		\Cmat_{\xvec_q} = & {\rm{E}}\left[ \xvec_q \xvec_q^H \right] = {\rm{E}}\left[ \mathcal{Q}(\yvec) \mathcal{Q}^H(\yvec) \right] \\
		=& \frac{2}{\pi} \left[\text{arcsin}\left( \left(\text{diag}\left(\Cmat_\yvec\right) \right)^{-\frac{1}{2}} \text{Re}(\Cmat_\yvec) \left( \text{diag} \left( \Cmat_\yvec \right) \right)^{-\frac{1}{2}} \right) \right. \\
		&+\left. j\text{arcsin} \left(\left(\text{diag} \left(\Cmat_\yvec\right) \right)^{-\frac{1}{2}} \text{Im}(\Cmat_\yvec) \left( \text{diag} \left( \Cmat_\yvec \right) \right)^{-\frac{1}{2}} \right) \right],
	\end{aligned}
\ee
where $\Cmat_\yvec$ is defined in \eqref{II.A.5}.
It should be noted that the matrix $\Cmat_{\xvec_q}$ in \eqref{VII.5} is well defined according to the following explanation. The elements of the matrix
\be \label{VII.6}
	\left[ \left(\text{diag} \left(\Cmat_\yvec\right) \right)^{-\frac{1}{2}} \Cmat_\yvec \left( \text{diag} \left( \Cmat_\yvec \right) \right)^{-\frac{1}{2}} \right]_{i,j} = \frac{\Cmat_{\yvec_i \yvec_j}}{\sqrt{\Cmat_{\yvec_i} \Cmat_{\yvec_j}}},
\ee
are, by definition, the Pearson correlation coefficients. Therefore, from the properties of the Pearson correlation coefficeint,
\be \label{VII.7}
	\left|\frac{\Cmat_{\yvec_i \yvec_j}}{\sqrt{\Cmat_{\yvec_i} \Cmat_{\yvec_j}}}\right| \leq 1,
\ee 
and since $\left|\text{Re}\left(\Cmat_{\yvec_i \yvec_j}\right)\right| \leq \left|\Cmat_{\yvec_i \yvec_j}\right|$ we obtain
\be \label{VII.8}
	\left|\frac{\text{Re}\left(\Cmat_{\yvec_i \yvec_j} \right)}{\sqrt{\Cmat_{\yvec_i} \Cmat_{\yvec_j}}}\right| \leq 1.
\ee
Similarly, the same holds for $\text{Im}\left(\Cmat_{\yvec_i \yvec_j}\right)$ and thus, the auto-covariance matrix $\Cmat_{\xvec_q}$ in \eqref{VII.5} is well defined.

Finally, to calculate the cross-covariance matrix of $\xvec_q$ with $\thetavec$ and $\xvec_a$ we use the Bussgang Theorem \cite{Bussgang}, which implies that for two zero-mean complex Gaussian random variables, $r$ and $t$, with 1-bit quantization as given in \eqref{I.1}, the cross-covariance is given by
\be \label{VII.9}
	\Cmat_{r \mathcal{Q}(t)} = {\rm{E}}\left[r \mathcal{Q}^*(t) \right] \sqrt{\frac{2}{\pi \sigma_t^2}} \Cmat_{r t},
\ee
where $\sigma_t^2$ is the covariance of the random variable $t$. Therefore, the cross-covariance matrix of $\thetavec$ and $\xvec_q$ 
is given by
\be \label{VII.10}
\begin{aligned}
	\Cmat_{\thetavecsmall \xvec_q} &= {\rm{E}}\left[ \thetavec \mathcal{Q}^H(\yvec) \right] \\
	&=\sqrt{\frac{2}{\pi}} {\rm{E}}\left[ \thetavec \left(\thetavec^H \Gmat^H + \wvec_q^H \right) \right] \left(\text{diag}\left(\Cmat_\yvec\right)\right)^{-\frac{1}{2}}\\
	&=\sqrt{\frac{2}{\pi}} \Sigmavec_\thetavecsmall \Gmat^H \left(\text{diag}\left(\Cmat_\yvec\right)\right)^{-\frac{1}{2}},
\end{aligned}
\ee
where the second equality is obtained by implementing the Bussgang formula from \eqref{VII.9} element-wise, and the last equality following the fact that $\thetavec$ and $\wvec_a$ are mutually independent and $\Gmat$ is deterministic. Similarly, from \eqref{VII.9} and due to the fact that $\wvec_a$, $\wvec_q$, and $\thetavec$ are mutually independent, the cross-covariance of $\xvec_a$ and $\xvec_q$ is given by
\be \label{VII.11}
	\begin{aligned}
		\Cmat_{\xvec_a \xvec_q} &= {\rm{E}}\left[ \xvec_a \: \mathcal{Q}(\yvec)^H \right] \\
		&= \sqrt{\frac{2}{\pi}} {\rm{E}}\left[ \left(\Hmat \thetavec + \wvec_a^H \right) \left(\thetavec^H \Gmat^H + \wvec_q^H \right) \right] \left(\text{diag}\left(\Cmat_\yvec\right)\right)^{-\frac{1}{2}} \\
		&= \sqrt{\frac{2}{\pi}} \Hmat \Sigmavec_\thetavecsmall \Gmat^H \left(\text{diag}\left(\Cmat_\yvec\right)\right)^{-\frac{1}{2}}.
	\end{aligned}
\ee
By substituting \eqref{VII.3}, \eqref{VII.5}, \eqref{VII.10}, and \eqref{VII.11} in \eqref{VII.1} and \eqref{VII.2} we obtain that the auto-covariance and cross-covariance matrices are given by \eqref{II.B.3} and \eqref{II.B.4}, respectively.

\section{Proof of Theorem \ref{Theorem}} \label{VIII}
In this appendix we prove that the LMMSE estimator from \eqref{II.B.1} and its MSE from \eqref{II.B.2} are reduced, under Assumptions A.1-A.3, to \eqref{III.A.5} and \eqref{III.A.6}, respectively.
By substituting $\Sigmavec_\thetavecsmall = \Imat_M$ from Assumption A.1 in \eqref{II.A.3}, \eqref{II.A.5}, and \eqref{VII.3} we obtain that in this case the auto-covariance and cross-covariance matrices satisfy
\be	\label{VIII.1}
	\Cmat_{\xvec_a} = \Hmat \Hmat^H + \sigma_a^2 \Imat_{N_a},
\ee
\be	\label{VIII.2}
	\Cmat_\yvec = \Gmat \Gmat^H + \sigma_q^2 \Imat_{N_q},
\ee
and
\be \label{VIII.3}
	\Cmat_{\thetavecsmall \xvec_a} = \Hmat^H.	
\ee
By substituting \eqref{III.A.3} and \eqref{III.A.4} from Assumption A.3 in \eqref{VIII.2}, we obtain
\be \label{VIII.4}
	\Cmat_\yvec=\rho_q (\onevec_{n_q} \onevec_{n_q}^T) \otimes \Imat_M + \sigma_q^2\Imat_{N_q},
\ee
which is a real matrix. Thus, by applying the diagonal operator on $\Cmat_\yvec$ in \eqref{VIII.4}, one obtains
\be \label{VIII.6}
	\text{diag}(\Cmat_\yvec) = (\rho_q + \sigma_q^2) \Imat_{N_q},
\ee
and therefore,
\be \label{VIII.7}
	\left(\text{diag}(\Cmat_\yvec)\right)^{-\frac{1}{2}} = \frac{1}{\sqrt{\rho_q + \sigma_q^2}} \Imat_{N_q}.
\ee
From \eqref{VIII.4} and \eqref{VIII.7}, we obtain that
\be \label{VIII.8}
	\begin{aligned}
		&\left(\text{diag}\left(\Cmat_\yvec\right) \right)^{-\frac{1}{2}} \Cmat_\yvec \left( \text{diag} \left( \Cmat_\yvec \right) \right)^{-\frac{1}{2}} = \\ &\hspace{2cm} \Imat_{N_q} + \frac{\rho_q}{\rho_q+\sigma_q^2} ((\onevec_{n_q} \onevec_{n_q}^T) \otimes \Imat_M - \Imat_{N_q}).
	\end{aligned}
\ee
Substituting \eqref{VIII.8} in \eqref{VII.5} we obtain
\be \label{VIII.9}
	\Cmat_{\xvec_q} = \frac{2}{\pi} \text{arcsin}\left( \left(\text{diag}\left(\Cmat_\yvec\right) \right)^{-\frac{1}{2}} \text{Re}(\Cmat_\yvec) \left( \text{diag} \left( \Cmat_\yvec \right) \right)^{-\frac{1}{2}} \right).
\ee
Applying the element-wise arcsin function on \eqref{VIII.8}, results in
\be	\label{VIII.10}
	\begin{aligned}
		\Cmat_{\xvec_q} &= \frac{2}{\pi} \left( \frac{\pi}{2} \Imat_{N_q} + \text{arcsin}\left( \frac{\rho_q}{\rho_q + \sigma_q^2}\right) \left( \frac{1}{\rho_q} \Gmat \Gmat^H - \Imat_{N_q} \right) \right) \\
		&= \alpha \Imat_{N_q} + (1-\alpha)\frac{1}{\rho_q} \Gmat \Gmat^H ,
	\end{aligned}
\ee
where $\alpha$ is defined in \eqref{III.A.7}.

Substituting $\Sigmavec_\thetavecsmall = \Imat_M$, from Assumption A.1, and \eqref{VIII.7} into 
\eqref{VII.11} and \eqref{VII.10} we obtain:
\be \label{VIII.11}
	\Cmat_{\xvec_a \xvec_q} = \sqrt{\frac{2}{\pi(\rho_q+\sigma_q^2)}} \Hmat \Gmat^H
\ee
and
\be \label{VIII.12}
	\Cmat_{\thetavecsmall \xvec_q} = \sqrt{\frac{2}{\pi(\rho_q+\sigma_q^2)}} \Gmat^H,
\ee
respectively. Substitution of \eqref{VIII.1}, \eqref{VIII.10}, \eqref{VIII.11}, and \eqref{VIII.12} in \eqref{II.B.3} and \eqref{II.B.4} results in
\be \label{VIII.13}
	\begin{aligned}
		\Cmat_\xvec &= \begin{bmatrix}
			\Cmat_{\xvec_a} & \Cmat_{\xvec_a \xvec_q} \\ \Cmat_{\xvec_q \xvec_a} & \Cmat_{\xvec_q}
		\end{bmatrix} \\ 
		&=\begin{bmatrix}
			 \Hmat \Hmat^H + \sigma_a^2 \Imat_{N_a} & \sqrt{\frac{2}{\pi(\rho_q+\sigma_q^2)}} \Hmat \Gmat^H\\
		 	\sqrt{\frac{2}{\pi(\rho_q+\sigma_q^2)}} \Gmat \Hmat^H & \alpha \Imat_{N_q} + (1-\alpha)\frac{1}{\rho_q} \Gmat \Gmat^H
		\end{bmatrix}
	\end{aligned}
\ee
and
\be	\label{VIII.14}
	\Cmat_{\thetavecsmall \xvec} = \begin{bmatrix}
		\Hmat^H & \sqrt{\frac{2}{\pi(\rho_q+\sigma_q^2)}} \Gmat^H
	\end{bmatrix},
\ee
respectively.

The auto-covariance matrix in \eqref{VIII.13} is a block matrix.
Therefore in order to calculate its inverse, we first note that using the Woodbury matrix identity (Eq. (0.7.4.1) \cite{horn2012matrix}) the inverse of the left upper block of $\Cmat_\xvec$, which is given in \eqref{VIII.1}, satisfies
\be \label{VIII.15}
	\begin{aligned}
		\Cmat_{\xvec_a}^{-1} &= \frac{1}{\sigma_a^2} \left( \Imat_{N_a} - \frac{1}{\sigma_a^2} \Hmat \left( \Imat_M + \frac{1}{\sigma_a^2} \Hmat^H \Hmat \right)^{-1} \Hmat^H \right) \\
		&=	\frac{1}{\sigma_a^2}\left( \Imat_{N_a} - \frac{1}{\rho_a n_a + \sigma_a^2} \Hmat \Hmat^H \right).
	\end{aligned}
\ee
where the last equality is obtained from \eqref{III.A.1} and \eqref{III.A.2} in Assumption A.2 which implies that $\Hmat^H \Hmat = \rho_a n_a \Imat_M$.
Second, by using \eqref{VIII.10}, \eqref{VIII.11}, and \eqref{VIII.15}, it can be verified that
\be \label{VIII.16}
	\begin{aligned}
		&\Cmat_{\xvec_q}- \Cmat_{\xvec_q \xvec_a}\Cmat_{\xvec_a}^{-1}\Cmat_{\xvec_a \xvec_q} = \alpha \Imat_{N_q} + (1-\alpha)\frac{1}{\rho_q} \Gmat \Gmat^H \\
		&\;- \frac{2}{\pi(\rho_q+\sigma_q^2)}\frac{1}{\sigma_a^2}\Gmat \Hmat^H \left( \Imat_{N_a} - \frac{1}{\rho_a n_a + \sigma_a^2} \Hmat \Hmat^H \right) \Hmat \Gmat^H \\
		&\; = \alpha \left(\Imat_{N_q} +\frac{\beta}{\alpha} \Gmat \Gmat^H \right) \define \Dmat,
	\end{aligned}
\ee
where $\beta$ is defined in \eqref{III.A.8} and using $\Hmat^H \Hmat = \rho_a n_a \Imat_M$.
By using the Woodbury matrix identity on \eqref{VIII.16}, the inverse matrix is given by:
\be \label{VIII.17}
	\begin{aligned}
		&\left(\Cmat_{\xvec_q}- \Cmat_{\xvec_q \xvec_a}\Cmat_{\xvec_a}^{-1}\Cmat_{\xvec_a \xvec_q}\right)^{-1} \\ &\hspace{0.5cm}= \frac{1}{\alpha} \left( \Imat_{N_q} -\frac{\beta}{\alpha}\Gmat \left(\Imat_M + \frac{\beta}{\alpha}\Gmat^H \Gmat\right)^{-1} \Gmat^H \right) \\
		&\hspace{0.5cm}= \frac{1}{\alpha} \left( \Imat_{N_q} - \frac{\beta}{\alpha + \beta \rho_q n_q} \Gmat \Gmat^H \right),
	\end{aligned}
\ee
where the last equality is obtained from \eqref{III.A.3} and \eqref{III.A.4} in Assumption A.3, which implies that $\Gmat^H \Gmat = \rho_q n_q \Imat_M$.
Using block matrix inversion together with the results in \eqref{VIII.15} and \eqref{VIII.17}, it can be verified that the inverse auto-covariance matrix in \eqref{VIII.13} is
\be \label{VIII.18}
	\begin{aligned}
		\Cmat_\xvec^{-1} &= \left[\begin{matrix*}[l] 
			\Cmat_{\xvec_a}^{-1}+\Cmat_{\xvec_a}^{-1}\Cmat_{\xvec_a \xvec_q} \Dmat^{-1} \Cmat_{\xvec_q \xvec_a} \Cmat_{\xvec_a}^{-1} &\vdots\\
			-\Dmat^{-1} \Cmat_{\xvec_q \xvec_a} \Cmat_{\xvec_a}^{-1} &\vdots 
		\end{matrix*} \right. \\
		&\hspace{3.9cm} \left. \begin{matrix*}[r] 
			-\Cmat_{\xvec_a}^{-1}\Cmat_{\xvec_a \xvec_q} \Dmat^{-1} \\
			\Dmat^{-1}
		\end{matrix*}\right]\\
		&=\left[\begin{matrix*}[l]
			\frac{1}{\sigma_a^2} \left( \Imat_{N_a} + \nu(n_a,n_q) \Hmat \Hmat^H \right) & \vdots \\
			-\xi(n_a,n_q) \Gmat \Hmat^H  &\vdots
		\end{matrix*} \right. \\
		&\hspace{2.5cm} \left.\begin{matrix*}[r]
			-\xi(n_a,n_q) \Hmat \Gmat^H  \\
			\frac{1}{\alpha} \left(\Imat_{N_q} - \frac{\beta(n_a)}{\alpha+ \beta(n_a) \rho_q n_q} \Gmat \Gmat^H\right)
		\end{matrix*}\right],
	\end{aligned}
\ee
where 
\be \label{VIII.19}
	\begin{aligned}
		&\nu(n_a,n_q) \define - \frac{1}{\rho_a n_a + \sigma_a^2} \\ &\hspace{0.5cm} +\frac{2 \rho_q n_q \sigma_a^2}{\pi (\rho_q+\sigma_q^2) (\alpha + \beta(n_a) \rho_q n_q ) (\rho_a n_a + \sigma_a^2)^2}
	\end{aligned}
\ee
and
\be \label{VIII.20}
	\xi(n_a,n_q) \define \sqrt{\frac{2}{\pi\left(\rho_q+\sigma_q^2\right)}} \frac{1}{(\alpha + \beta(n_a) \rho_q n_q ) (\rho_a n_a + \sigma_a^2)}.	
\ee
Substituting \eqref{VIII.14} and \eqref{VIII.18} into \eqref{II.B.1} and \eqref{II.B.2} we obtain \eqref{III.A.5} and \eqref{III.A.6}, respectively.

\section {Proof of Proposition \ref{Proposition}} \label{IX}
In this appendix we prove that for a given number of analog measurements, $n_a$, taking the maximum number of quantized measurements, $n_q$, under the power constraint is optimal in terms of minimizing the MSE, as in the optimization problem given in \eqref{III.C.1}.
We show that for any given value of $n_a$, we obtain that
\be \label{IX.1}
	MSE\vert_{n_a,n_q} - MSE\vert_{n_a,n_q+1} \geq 0,
\ee
meaning that the addition of a quantized measurement can only improve the MSE. Substituting the MSE in \eqref{III.A.6} into \eqref{IX.1}, we obtain that
\be \label{IX.2}
	\begin{aligned}
		& MSE\vert_{n_a,n_q} - MSE\vert_{n_a,n_q+1} \\
		&=M - M\left(\frac{\rho_a n_a }{\rho_a n_a + \sigma_a^2} \right.\\
		&\hspace{0.6cm} \left. +\frac{2\rho_q n_q \sigma_a^4}{\pi (\rho_q+\sigma_q^2) \left(\alpha + \beta(n_a)\rho_q n_q\right) \left(\rho_a n_a + \sigma_a^2\right)^2}\right) \\
		&\hspace{0.1cm}- \left[M - M\left(\frac{\rho_a n_a }{\rho_a n_a + \sigma_a^2} \right.\right.\\
		&\hspace{0.6cm} \left.\left. +\frac{2\rho_q (n_q+1) \sigma_a^4}{\pi (\rho_q+\sigma_q^2) \left(\alpha + \beta(n_a)\rho_q (n_q+1)\right) \left(\rho_a n_a + \sigma_a^2\right)^2}\right)\right]\\
		&= M \frac{2 \rho_q \sigma_a^4}{\pi (\rho_q + \sigma_q^2) (\rho_a n_a + \sigma_a^2)^2} \left[ \frac{n_q + 1}{\alpha +\beta(n_a) \rho_q(n_q+1)}\right. \\
		&\hspace{0.6cm} \left. - \frac{n_q}{\alpha +\beta(n_a) \rho_q n_q}\right]\\
		&= M \frac{2 \rho_q \sigma_a^4}{\pi (\rho_q + \sigma_q^2) (\rho_a n_a + \sigma_a^2)^2} \\
		&\hspace{0.9cm} \cdot \frac{\alpha}{(\alpha +\beta(n_a) \rho_q (n_q+1))(\alpha +\beta(n_a) \rho_q n_q)} \geq 0.
	\end{aligned}
\ee
In the following, we show that $\beta(n_a) > 0$. First, it can be seen that according to the definition of $\beta$ in \eqref{III.A.8} we have
\be \label{IX.3}
	\begin{aligned}
		\beta(n_a) &= \frac{2}{\pi\rho_q} \arcsin\left( \frac{\rho_q}{\rho_q+\sigma_q^2} \right) - \frac{2}{\pi\rho_q} \frac{\rho_a n_a}{\rho_a n_a + \sigma_a^2} \frac{\rho_q}{\rho_q +\sigma_q^2}\\
		&= \frac{2}{\pi\rho_q} \left( \arcsin\left( \frac{\rho_q}{\rho_q+\sigma_q^2} \right) - \frac{\rho_a n_a}{\rho_a n_a + \sigma_a^2} \frac{\rho_q}{\rho_q +\sigma_q^2} \right),
	\end{aligned}
\ee
where
\be \label{IX.4}
	0 \leq \frac{\rho_a n_a}{\rho_a n_a + \sigma_a^2} \leq 1,
\ee
and since $\rho_a>0$, $\sigma_a^2\geq 0$, and $n_a\geq 0$. In addition,
\be \label{IX.5}
	0< \frac{\rho_q}{\rho_q + \sigma_q^2} \leq 1,
\ee
since $\rho_q>0$ and $\sigma_q^2\geq 0$.
Therefore, due to \eqref{IX.4} and \eqref{IX.5}, the following expression is also positive
\be \label{IX.6}
	\arcsin\left( \frac{\rho_q}{\rho_q+\sigma_q^2} \right) - \frac{\rho_a n_a}{\rho_a n_a + \sigma_a^2} \frac{\rho_q}{\rho_q +\sigma_q^2} > 0,
\ee
and thus, we can conclude that $\beta(n_a)>0$.
Moreover, as a result of \eqref{IX.5}, we also have that
\be \label{IX.7}
	0 \leq \alpha < 1.
\ee
Since $n_a$, $n_q$, $\rho_a$, $\rho_q$, $\sigma_a^2$, $\sigma_q^2$, $\alpha$, and $\beta(n_a)$ in \eqref{IX.2} are non-negative, the inequality in \eqref{IX.1} holds.
Therefore, given the number of analog measurements, we take the maximum number of quantized measurements possible under the power constraint.
This in turn allows us to solve using a one-dimensional search over $n_a$ with the value of $n_q$ given in \eqref{III.C.3}.

\end{document}